\documentclass[submission,copyright,creativecommons]{eptcs}
\providecommand{\event}{MEMICS 2016} % Name of the event you are submitting to
\usepackage{breakurl}             % Not needed if you use pdflatex only.
\usepackage{underscore}           % Only needed if you use pdflatex.
\usepackage[utf8]{inputenc}
\usepackage{amsmath}
\usepackage{amsthm}
\usepackage{amsfonts}
\usepackage{amssymb}
\usepackage{tikz}
\usepackage{verbatim}
\usepackage{appendix}
\usetikzlibrary{graphs}
\usetikzlibrary{arrows}
\usepackage{algorithm}
\usepackage[square,numbers]{natbib}
\usepackage{makeidx}      %% The `makeidx` package contains

\title{Characterizing DAG-depth of Directed Graphs}
\author{Matúš Bezek\thanks{Research supported by the Center of Excellence --
Institute for Theoretical Computer Science;
Czech Science Foundation Project no.~P202/12/G061.}
\institute{Faculty of Informatics, \\Masaryk University, \\Brno, Czech Republic,\\
\email{xbezek@fi.muni.cz}}}
\def\titlerunning{DAG-depth}
\def\authorrunning{M. Bezek}
\begin{document}
\maketitle

\begin{abstract}
    We study DAG-depth, a structural depth measure of directed graphs, 
	which naturally extends the tree-depth of ordinary graphs. We define a DAG-depth 
    decomposition as a strategy for the cop player in the lift-free version of the
    cops-and-robber game on directed graphs and prove its correctness. The DAG-depth
    decomposition is related to DAG-depth in a similar way as an elimination tree is related
    to the tree-depth. We study the size aspect of DAG-depth decomposition and provide a 
    definition of mergeable and optimally mergeable vertices in order to make the decomposition
    smaller and acceptable for the cop player as a strategy. We also provide a way to find the closure
    of a DAG-depth decomposition, which is the largest digraph for which the given decomposition 
    represents a winning strategy for the cop player.
\end{abstract}

%% We will define several mathematical sectioning commands.
\newtheorem{theorem}{Theorem}[section] %% The numbering of theorems
                               %% will be reset after each section.
\newtheorem{lemma}[theorem]{Lemma}     %% The numbering of lemmas
\newtheorem{corr}[theorem]{Corrolary}  %% and corrolaries will
                                %% share the counter with theorems.
%%\theoremstyle{definition}
\newtheorem{definition}{Definition}
\newtheorem{observation}[theorem]{Observation}
\newtheorem{proposition}[theorem]{Proposition}
\newtheorem{innercustomthm}{Theorem}
\newenvironment{customthm}[1]
  {\renewcommand\theinnercustomthm{#1}\innercustomthm}
  {\endinnercustomthm}
\newtheorem{innercustomlemma}{Lemma}
\newenvironment{customlemma}[1]
  {\renewcommand\theinnercustomlemma{#1}\innercustomlemma}
  {\endinnercustomlemma}

\section{Introduction}
Structural width parameters are numeric parameters associated with graphs. They represent different
properties of graphs. Examples of such parameters are path-width~\cite{article-minorspath}, 
tree-width~\cite{article-minorstree} and clique-width~\cite{article-cliquewidth}. The first two 
were defined by Robertson and Seymour in 1980s, clique-width was defined by Courcelle~et~al. 
in 1991. Informally, path-width measures how close a graph is to a path and the other two
similarly relate to trees.

As a directed analog of tree-width, directed tree-width~\cite{article-dirtreewidth} was defined
by Johnson et al. in 1998. This line of research continued in Obdržálek's definition of 
DAG-width~\cite{article-dagwidth} in 2006. Another digraph measure 
Kelly-width~\cite{article-kellywidth} was defined by Hunter and Kreutzer in 2007. In 2010 
Ganian~et~al. analyzed~\cite{article-gooddigraph} digraph width measures and reasons why the 
search for the "perfect" directed analog of tree-width has not been successful so far.

All these parameters are tightly correlated with different versions of a cops-and-robber game with
an infinitely fast robber. The essence of this game is to catch the robber by placing 
the cops in the vertices and moving them. 

Structural depth parameters are analogously correlated with the so-called lift-free version of
the game, defined in Section~\ref{sec:liftfree}. An example of such a parameter is 
tree-depth~\cite{article-treedepth}, defined by Nešetřil and Ossona de Mendez in 2005.
In 2012 Adler~et~al. defined~\cite{article-hypertree} a hypergraph analog of tree-depth.
In that work Adler~et~al. also studied generalizations of the elimination tree
for hypergraphs.

A directed analog of tree-depth was defined under the name DAG-depth~\cite{article-ddp} by Ganian et al. in
2009. Its definition, however, did not provide any structural insight into the parameter since
there was no naturally associated decomposition with it.

We define a DAG-depth decomposition of a digraph and show that it can be used as a winning
strategy for the cop player in the lift-free version of the cops-and-robber game in directed
graph. The main issue is that an optimal decomposition usually has to contain multiple copies
of original vertices.

\section{Preliminaries}
\label{chp:definitions}
\label{chp:game}
\newcommand{\cupdot}{\mathbin{\mathaccent\cdot\cup}}
  We deal with directed graphs.

  An \textit{outdegree} $d_D^+(v)$ of the vertex $v \in V(D)$ is the number of edges going
  from $v$. An \textit{indegree} $d_D^-(v)$ of the vertex $v \in V(D)$ is the number of edges
  coming to $v$. An \textit{out-neighborhood}, denoted by $N_D^+(v)$, is the set of vertices $x$ 
  such that the edge $(v,x)$ exists in $D$.
  
  An acyclic directed graph is shortly called a DAG.
  In DAG, vertex $u$ is a parent of $v$ if the digraph contains an edge
  $(u,v)$. Vertex $v$ is then a child of $u$. The vertex $u$ is an ancestor of $v$ if the
  digraph contains a path from $u$ to $v$. If $u$ is an ancestor of $v$, then $v$ is a 
  descendant of $u$.
  
  One of the ways to extend connectivity to directed graphs is the concept of \textit{reachable
  fragments}. Reachable fragments are maximal, by inclusion, sets of vertices such that every
  fragment $R$ contains a vertex called the source, from which there is a path to every vertex 
  of $R$.

\subsection{Original cops-and-robber game}
The \textit{cops-and-robber} game was first introduced~\cite{article-lapaugh} in 1982 by LaPaugh.
The variant we are interested in was introduced~\cite{article-orgtwidth} in 1989 by Seymour and 
Thomas. The main difference between them is that in the version by Seymour and Thomas the robber
is infinitely fast, while in LaPaugh's version he is not.

The game by Seymour and Thomas is played by one player on a~finite undirected graph $G$. The 
player controls $k$ cops. At any time each of them either stays on some vertex or is temporarily
removed from the graph. The player moves the cops, he can remove them from the graph and place 
them back into any vertex he wants.

The robber always stands on some vertex of $G$. During the game, he can move at any time along 
the edges at infinite speed. He is not allowed to run through a cop but he can see when the cop
is being placed on some vertex and he can run through that vertex before the cop lands.

The cop player wins when the cops catch a robber, i.e. when the robber is in some vertex $v$ 
such that there is a cop placed in each vertex of $N^+(v)$ and also on $v$. The player 
loses if the robber is able to avoid getting caught.

The robber is always aware of cops' position and the player is aware of robber's position.
The minimal number of cops needed to catch a robber on a graph is called the \textit{cop number}
of the graph.

\subsection{Lift-free version of the game}
\label{sec:liftfree}
In the \textit{lift-free} version of the cops-and-robber game there is an additional rule; once 
the cop is placed to some vertex, he stays there until the end of the game. In each turn the cop
player puts a cop onto some vertex of the graph. The game ends when the robber is caught or the
cop player runs out of cops. If the robber is caught, the cop player wins, otherwise he loses.

%The cop number can be naturally extended to this version of the game.

\subsection{Extension to directed graphs}
The concept of the cops-and-robber game can be naturally extended to directed graphs. 
The robber can only move along the edges in the right direction.

The aforementioned DAG-depth~\cite{article-ddp}, introduced by Ganian et al. in 
2009, is given as follows.
\begin{definition}[DAG-depth]
  Let $D$ be a digraph and $R_1,\dots,R_p$ the reachable fragments of $D$. The DAG-depth $ddp(D)$ 
  is inductively defined:
  \[ ddp(D) =
  \begin{cases}
   1										& \quad \text{if } |V(D)|=1\\
   1 + \min_{v \in V(D)} (ddp(D - v))		& \quad \text{if } p=1 \text{ and } |V(D)|>1\\
   \max_{1 \leq i \leq p} ddp(R_i) 			& \quad \text{otherwise}\\
  \end{cases}
  \]
\label{def:ddp}
\end{definition}

DAG-depth is directly related to the lift-free game as follows (where a proof is quite obvious):
\begin{theorem}
  Let $D$ be a digraph. There exists a lift-free winning strategy for $k$ cops, if and only if
  DAG-depth of $D$ is less or equal to $k$.
\end{theorem}

\section{DAG-depth decomposition}
\label{chp:decomposition}
The aim of this section is to define a DAG-depth decomposition (Definition~\ref{def:nCover}) from 
the recursive definition of DAG-depth (Definition~\ref{def:ddp}) the same way as an elimination tree
is obtained from the definition of tree-depth. The decomposition aims to represent a game plan for 
the cop player.

The main difference between the tree-depth and DAG-depth cases is that in undirected graphs two 
connected components cannot have any vertices in common, while distinct maximal reachable fragments
in directed graphs can have some vertices in common. This naturally brings complications and 
ambiguity.

There could be two ways to resolve this. Either just ignore it and let the decomposition have more 
copies of one vertex. But that would mean the decomposition could grow exponentially large (see 
Section~\ref{sec:explarge}, and
exponentially large decompositions would be practically useless for the player as a game plan and 
for algorithmic purposes.

\begin{figure}[H]
\begin{center}
\begin{tikzpicture}
\tikzset{vertex/.style = {shape=circle,draw,minimum size=1.5em}}
\tikzset{edge/.style = {->,> = latex', line width=1.4pt}}
% vertices
\node[vertex] (a) at  (4,3) {A};
\node[vertex] (b) at  (4,1) {B};
\node[vertex] (c) at  (5.5,2) {C};
\node[vertex] (d) at  (7.5,2) {D};
\node[vertex] (e) at (9,3) {E};
\node[vertex] (f) at (9,1) {F};
\node[vertex] (da) at  (4,-4) {A};
\node[vertex] (db) at  (5,-4) {B};
\node[vertex] (dc) at  (5,-1) {C};
\node[vertex] (dd) at  (6,-4) {D};
\node[vertex] (dc2) at  (7,-4) {C};
\node[vertex] (dd2) at  (8,-1) {D};
\node[vertex] (de) at (8,-4) {E};
\node[vertex] (df) at (9,-4) {F};
%edges
\draw[edge] (a) to (c);
\draw[edge] (b) to (c);
\draw[edge] (c) to (d);
\draw[edge] (d) to (c);
\draw[edge] (e) to (d);
\draw[edge] (f) to (d);
\draw[edge] (dc) to (da);
\draw[edge] (dc) to (db);
\draw[edge] (dc) to (dd);
\draw[edge] (dd2) to (de);
\draw[edge] (dd2) to (df);
\draw[edge] (dd2) to (dc2);
\end{tikzpicture}
\end{center}
\caption{A simple digraph and its decomposition, showing that DAG-depth decomposition cannot
always be done optimally without repetition of vertices (see repeated $C$,$D$)}
\label{fig:mygraph}
\end{figure}
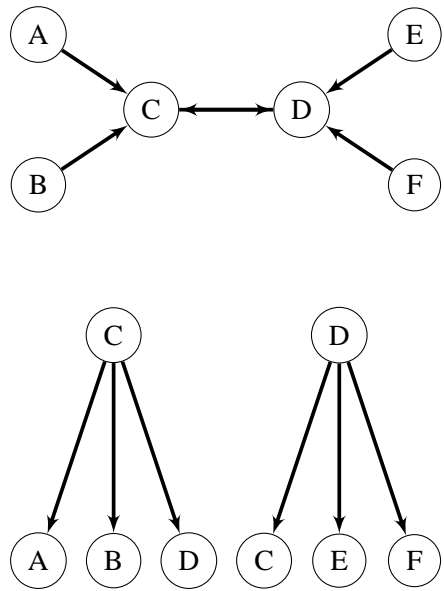

The other extreme solution would be to merge all the copies of one vertex. This cannot always 
be done, as the graph in Figure~\ref{fig:mygraph} shows.

In this graph, the robber can be caught by using two cops. The idea is that if the robber 
starts on the vertices $A$ or $B$, the player places the first cop on the vertex $C$. Then
the robber has either stayed on the vertex he was on, or ran to $D$. Placing the cop on the
robber will catch him, since there is no edge between $A$ and $D$ or $B$ and $D$. Symmetrically,
if the robber starts on $E$ or $F$, the first cop is placed to $D$ and second catches the robber. 
If the robber starts on $C$ or $D$, he can not go into any other vertex and covering these two 
vertices in any order will result into a win.

In the corresponding "decomposition" (Figure~\ref{fig:mygraph} bottom), the two copies 
of the vertex $C$ can not be merged since their merging would create a path of length 
$2$ and therefore the decomposition would not be optimal anymore. For the same reason 
the copies of $D$ can not be merged, too.

Balancing these two extreme approaches would give us decompositions with some of the repeated 
vertices merged. Now the core question is, which vertices can be merged and why.

\subsection{Basic properties of a DAG-depth decomposition}
As argued above, in a decomposition some vertices will be copies of the same original 
vertex $v \in D$. To properly work with this fact, there is a need to formally distinguish 
the two vertex sets and map between them.

This is the first difference from an elimination tree of tree-depth where the vertex sets are
identical. Formally, this can be done by defining the function $org : V(P) \to V(D)$ which takes
a vertex from the decomposition and returns its original from the digraph~$D$. Vertices $x,y \in
V(P)$ are copies of the same vertex if and only if $org(x) = org(y)$.\\

\textit{Roots} of a DAG are all of its vertices whose indegree is zero. Vertices whose outdegree
is zero are called \textit{leafs}.

The \textit{level} of a vertex in a DAG is the maximal length of a directed path from a root to this
vertex. The \textit{depth} of a vertex is the maximal length of a path from this vertex to a leaf.
The \textit{depth} of a DAG is the maximum depth over its vertices.

\begin{definition}[DAG-depth decomposition]
  A \textit{DAG-depth decomposition} of a digraph $D$ is a DAG $P$ and a surjective function 
  $org : V(P) \to V(D)$. Furthermore, a DAG-depth decomposition is called \textit{valid} if 
  the following Neighbor cover condition holds.
  
  The \textit{Neighbor cover} condition states that for every vertex $v' \in V(P)$ such that 
  $org(v') = v$, the following holds:\\
  For every $u \in N_D^+(v)$, 
  \begin{enumerate}
   \item there exists $u' \in V(P)$ such that $org(u')=u$ and $u'$ is a descendant of $v'$ 
   in $P$, or
   \item every path from any root of $P$ to $v'$ contains a vertex $u' \in V(P)$ such that
   $org(u')=u$.
  \end{enumerate}
\label{def:nCover}
\end{definition}

Let $P$ be a DAG-depth decomposition of some graph $D$ such that the depth of $P$ is equal
to the DAG-depth of $D$. $P$ is then called an \textit{optimal} decomposition. Such 
decomposition exists for any digraph $D$, as Theorem~\ref{thm:validdecK} shows.\\

The following example of Figure~\ref{gra:hlinGraph} illustrates how a valid DAG-depth decomposition
can be used as a strategy for the cop player to catch the robber.

If the player is to use the decomposition in Figure~\ref{gra:hlinGraph} as a game plan, he 
starts by covering the vertex $E$, since its copy is the only root.

Then, if the robber is in the vertex $A$ or $B$, the player continues by covering the vertex
$B$. If the robber was in this vertex, he has been caught. Otherwise he is in the vertex $A$,
which the player will cover by the third cop and therefore catch the robber.

If the robber was not in the vertex $A$ or $B$, the player places the second cop in the vertex
$G$, whose copy is on the same level as the copy of $B$. There are now three possibilities
where the robber can be. The first one is that he is in one of the vertices $C$, $D$. The
second one is that he is in the vertex $F$. The last one is that the robber is hiding in one
of the vertices $H$, $I$, or $J$.

\begin{figure}
\begin{center}
\begin{tikzpicture}
\tikzset{vertex/.style = {shape=circle,draw,minimum size=1.5em}}
\tikzset{edge/.style = {->,> = latex', line width=1.4pt}}
% vertices
\node[vertex] (a) at  (-2.5,4) {A};
\node[vertex] (b) at  (-1,4) {B};
\node[vertex] (c) at  (6.5,4) {C};
\node[vertex] (d) at  (5,4) {D};
\node[vertex] (e) at (0.5,4) {E};
\node[vertex] (f) at (2,4) {F};
\node[vertex] (g) at  (3.5,4) {G};
\node[vertex] (h) at  (3.5,2.5) {H};
\node[vertex] (i) at (2,2.5) {I};
\node[vertex] (j) at (0.5,2.5) {J};
\node[vertex] (da) at  (0,-1.25) {A};
\node[vertex] (db) at  (1,0) {B};
\node[vertex] (dc) at  (5,-2.25) {C};
\node[vertex] (dd) at  (4,-1.25) {D};
\node[vertex] (de) at (2,1) {E};
\node[vertex] (df) at (3,-1.25) {F};
\node[vertex] (dg) at  (3,0) {G};
\node[vertex] (dh) at  (1,-2.25) {H};
\node[vertex] (di) at (2,-1.25) {I};
\node[vertex] (dj) at (3,-2.25) {J};
%edges
\draw[edge] (a) to (b);
\draw[edge] (b) to (e);
\draw[edge] (c) to (d);
\draw[edge] (d) to (g);
\draw[edge] (e) to (f);
\draw[edge] (f) to (e);
\draw[edge] (g) to (f);
\draw[edge] (f) to (g);
\draw[edge] (g) to (h);
\draw[edge] (h) to (g);
\draw[edge] (i) to (h);
\draw[edge] (h) to (i);
\draw[edge] (i) to (j);
\draw[edge] (j) to (i);
\draw[edge] (e) to (j);
\draw[edge] (j) to (e);
\draw[edge] (de) to (dg);
\draw[edge] (de) to (db);
\draw[edge] (db) to (da);
\draw[edge] (dg) to (df);
\draw[edge] (dg) to (di);
\draw[edge] (dg) to (dd);
\draw[edge] (dd) to (dc);
\draw[edge] (di) to (dj);
\draw[edge] (di) to (dh);

\end{tikzpicture}
\end{center}
\caption{A simple graph and its valid and optimal decomposition}
\label{gra:hlinGraph}
\end{figure}
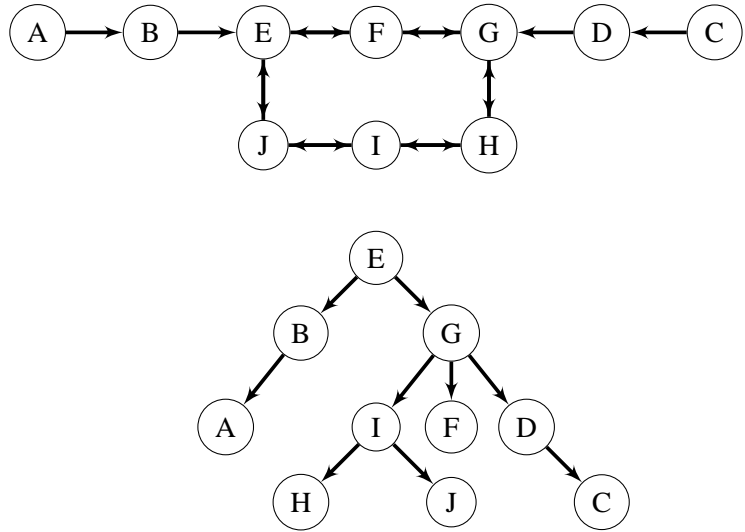

If it is the first case, the player continues by covering the vertex~$D$, whose copy is the child
of the copy of the last covered vertex. If the robber was here, he is caught, otherwise the cop
is placed to $C$ and catches the robber.

If it is the second case and the robber is hiding in the vertex $F$, the player simply covers
the vertex $F$ and catches the robber. If the robber is in one of the vertices $H$, $I$, or $J$,
the player covers the vertex $I$, whose copy is the child of the copy of the last covered vertex 
$G$. The robber then escapes either to vertex $H$ or to $J$. In the last step the player simply 
covers the vertex robber is in and catches him.

The game rules outlined in this example are formally defined next.

\begin{definition}
  Given a DAG-depth decomposition $(P,org)$ of a digraph $D$, the cop player's strategy is as
  follows. The cops are placed on the vertices of $D$ and every cop is placed "because of" some
  vertex of $P$. The following convention is observed: if we say a cop is to be placed to a 
  vertex $v' \in V(P)$, he is actually placed to $v \in V(D)$ such that $org(v')=v$, unless the
  vertex has been covered by another cop before. In that case, no cop is placed in this step. 
  Then, the strategy based on $(P,org)$ is described by two simple rules:
  \begin{enumerate}
   \item The first cop is always placed to one of the roots of $P$. Each subsequent cop is 
   placed to the out-neighborhood of the previous cop in $P$.
   \label{num:placingRule1}
   \item Among the possible positions from \ref{num:placingRule1}, the actually chosen one 
   must have in $P$ a directed path to a copy of robber's current position.
   \label{num:placingRule2}
  \end{enumerate}
\label{def:placingRules}
\end{definition}

The choice of the next vertex to be covered by a cop in Definition~\ref{def:placingRules} is generally 
non-deterministic, since more vertices can contain robber's position as a~descendant. 

\begin{theorem}
  Let $D$ be a digraph and $(P,org)$ its DAG-depth decomposition of depth $k$. Then the decomposition
  is valid if and only if every strategy based on $(P,org)$ by Definition~\ref{def:placingRules} is
  winning for $k$ cops.
\end{theorem}
\begin{proof}
  $(\Leftarrow$) The decomposition is valid if the Neighbor cover condition holds by 
  Definition~\ref{def:nCover}. Suppose the contrary: there exist a pair of vertices $u,v \in V(D)$
  such that edge $(u,v) \in E(D)$ exists. Also a vertex $u' \in V(P)$ exists such that $org(u')=u$ 
  and $u'$ does not contain any copy of the vertex $v$ as a descendant. Since none of the conditions 
  in the definition of the Neighbor cover condition holds, there exists a~path $p$ from some root 
  to $u'$ such that $p$ does not contain any copy of the vertex $v$.
  
  Let the player use the decomposition according to rules specified in Definition~\ref{def:placingRules}.
  If the robber started on the vertex $u \in V(D)$, then the player could proceed along the path
  $p$, since all of its vertices contain the copy $u' \in V(P)$ as a descendant. When the player 
  covers the vertex $u$ because of $u'$, the robber can escape to the vertex $v \in V(D)$ since 
  the path $p$ does not contain any copy of that vertex and therefore it is not covered by a cop. 
  Since the vertex $u'$ does not contain any copy of $v$ as a descendant, the player playing 
  according the Definition~\ref{def:placingRules} can not cover $v$ and the robber wins. The 
  given decomposition therefore represents a strategy which is not winning.
  
  $(\Rightarrow$) The other direction is the subject of subsequent claims and will follow from 
  Theorem~\ref{thm:validplan}.
  
\end{proof}

\begin{theorem}
  If the DAG-depth of a digraph $D$ is $k$, then there exists a~valid DAG-depth decomposition
  of $D$ of depth $k$.
\label{thm:validdecK}
\end{theorem}
\begin{proof}
  If $|V(D)|=1$, then $ddp(D)=1$. A decomposition with depth one exists, since it consists also 
  of the only vertex.\\
  A decomposition that consists of one vertex is always valid, since the original graph
  did not contain any edges and the Neighbor cover condition therefore always holds.
  
  If $|V(D)|>1$ and $D$ consists of the only reachable fragment, then the DAG-depth is computed
  as $ddp(d)=1 + \min_{v \in V(D)} (ddp(D - v))$. Such vertex $v$ is then the root of the 
  decomposition and is connected to the roots of the recursive decomposition of the rest of a 
  graph.\\
  Since the vertex $v$ was chosen to be the root, all other vertices are its descendants. 
  Therefore all the vertices of its out-neighborhood are his descendants, and for the rest of
  the graph the Neighbor covercondition holds by induction. That means the decomposition is
  valid.
  
  Otherwise, $D$ consists of more reachable fragments. The decomposition of each of them can
  be computed separately. When a disjoint union of them is made to a single graph, its depth
  will be equal to the maximum of the decompositions of the fragments. This is in accordance
  with Definition~\ref{def:ddp}.\\
  Since the decomposition is a disjoint union of the decompositions of the fragments, by induction
  the decompositions of the fragments are valid. Therefore their disjoint union is a valid 
  decomposition too.

\end{proof}

\begin{theorem}
  Let $D$ be a digraph for which there exists a valid DAG-depth decomposition of depth $k$. Then, 
  any strategy observing the rules of Definition~\ref{def:placingRules} is a winning
  strategy for at most $k$ cops.
\label{thm:validplan}
\end{theorem}
\begin{proof}
  A decomposition $(P,org)$ is valid if the Neighbor cover condition from Definition~\ref{def:nCover}
  holds. The cop player wins when the cop is placed on top of the robber to a vertex $r$ and all
  vertices from $N^+_D(r)$ are already covered by the cops.
  
  Let the last move of the cop be to vertex $u \in V(D)$, because of its copy $u' \in V(P)$ as in
  Definition~\ref{def:placingRules}. We claim that the robber may move only to vertices of $D$ 
  whose copies in $P$ are reachable from $u'$ in $P$.
  
  Before the robber moves, the previous statement holds because of the rule~\ref{num:placingRule2}
  of Definition~\ref{def:placingRules}.
  
  Let the robber be on a vertex $r \in V(D)$ and $r' \in V(P)$ its copy such that it is a descendant 
  of $u'$. The statement still holds if the robber moves along an edge $(r,v) \in E(D)$ to vertex 
  $v \in V(D)$ which has not yet been covered by a cop. From 
  Definition~\ref{def:nCover} we know that in the decomposition $P$ every path from a root to $r'$
  contains a copy of $v$ or $P$ contains a vertex $v' \in V(P)$ such that it is a descendant of $r'$.
  If $v'$ is a~descendant of $r'$, it is also a descendant of $u'$ since $r'$ is its descendant. 
  If every path from a root to $r'$ contains a copy of $v$, such copy must be a descendant of $u'$.
  If it was not, then $v'$ would have to lie on some path from a root to $u'$, and $v$ would have
  already been covered by a~cop by Definition~\ref{def:placingRules}.
  
  The previous invariant allows the player to always fulfill the rules of 
  Definition~\ref{def:placingRules}.
  
  The rules in Definition~\ref{def:placingRules} end with covering a vertex from $V(D)$ because of
  its copy which is a leaf in decomposition $P$. That means that all the neighbors of the covered 
  vertex have been covered before and the robber is caught. The decomposition therefore represents 
  a winning strategy.
  
  In every move, the vertex $v \in V(D)$ is covered because of some $v' \in V(P)$. Such vertex $v'$
  is always a child of the previous $v'$ and therefore all such vertices create a path in $P$. If 
  the player used more than $k$ cops, the path would need to be longer than $k$. Since the depth of
  $P$ is $k$, such path can not exist. Therefore, the decomposition represents a strategy for at most
  $k$ cops.

\end{proof}

\section{Merging the copies}
\label{sec:explarge}
We now return to the size aspect of a DAG-depth decomposition (say, the one obtained by 
Theorem~\ref{thm:validdecK}). We start with an example that it could be exponentially large.
To reduce the size of the decomposition, some copies of 
the same vertex should be merged while preserving validity of the decomposition. Not all
vertices with the same original can be merged (recall Figure~\ref{fig:mygraph}, though.

\begin{theorem}
  There exists a digraph $D$ such that its only valid and optimal DAG-depth decomposition without
  merging any vertices is exponentially large.
\label{thm:expLarge}
\end{theorem}
\begin{proof}
  Let $D$ be a digraph such that $V(D) = \{a_1, a_2, \dots, a_n, b_1, b_2, \dots,$ $b_n\}$ for
  $n \in \mathbb{N}$ and $E(D) = \{(a_i,a_j)$ $\cup (a_i,b_j) \cup (b_i,a_j) \cup (b_i,b_j)\}$ for 
  every $1 \leq i < j \leq n$. See Figure~\ref{gra:expo}.

  The digraph $D$ consists of two isomorphic reachable fragments – $V(D) \setminus \{a_1\}$ and 
  $V(D) \setminus \{b_1\}$.
  
  In the reachable fragment $V(D) \setminus \{b_1\}$ the only optimal first move of the cop player is 
  placing the cop onto $a_1$ since if he made any other move, one of the subgraphs 
  $\{a_1,a_2, \dots, a_n\}$ and $\{a_1,b_2, \dots, b_n\}$ is left uncovered. These subgraphs each 
  require another $n$ cops to catch the robber, while the DAG-depth of $D$ is $n = 1 + n - 1$.
  
  After covering $a_1$, the remaining digraph has the vertex set $\{a_2, \dots,$ $a_n, b_2, \dots,$ 
  $b_n\}$. Its decomposition can be found by induction.
  
  The reachable fragment $V(D) \setminus \{a_1\}$ is isomorphic to $V(D) \setminus \{b_1\}$ and 
  therefore the only optimal move is to cover $b_1$ and the remaining digraph is the same as for
  the first reachable fragment.
  
  The decomposition of $D$ will thus contain decomposition of the remaining digraph two times – as a 
  descendant of $a_1$ and as a descendant of $b_1$. The decomposition therefore consists of 
  $\sum_{i=1}^n 2^i = 2^{n+1}-2$ vertices.

\end{proof}

\begin{figure}[H]
\begin{center}
\begin{tikzpicture}

\tikzset{vertex/.style = {shape=circle,draw,minimum size=1.5em}}
\tikzset{edge/.style = {->,> = latex', line width=1.4pt}}
\tikzset{edgea/.style = {-,> = latex', line width=1.4pt}}
% vertices
\node[vertex] (a1) at (0,3) {$a_1$};
\node[vertex] (a2) at (1.5,3) {$a_2$};
\node[vertex] (a3) at (3,3) {$a_3$};
\node[vertex] (an) at (6,3) {$a_n$};
\node[vertex] (b1) at (0,0) {$b_1$};
\node[vertex] (b2) at (1.5,0) {$b_2$};
\node[vertex] (b3) at (3,0) {$b_3$};
\node[vertex] (bn) at (6,0) {$b_n$};
\node at (4.5,3) {\ldots};
\node at (4.5,0) {\ldots};
\node (anb1) at (5.68,2.75) {};
\node (anb2) at (5.6,2.78) {};
\node (bna1) at (5.67,0.25) {};
\node (bna2) at (5.6,0.22) {};
\node[vertex] (da1) at  (0,-1.5) {$a_1$};
\node[vertex] (db1) at  (6,-1.5) {$b_1$};
\node[vertex] (da21) at  (-1,-3) {$a_2$};
\node[vertex] (db21) at  (1,-3) {$b_2$};
\node[vertex] (da22) at  (5,-3) {$a_2$};
\node[vertex] (db22) at  (7,-3) {$b_2$};
\node[vertex] (da31) at  (-1.5,-4.5) {$a_3$};
\node[vertex] (db31) at  (-0.5,-4.5) {$b_3$};
\node[vertex] (da32) at  (0.5,-4.5) {$a_3$};
\node[vertex] (db32) at  (1.5,-4.5) {$b_3$};
\node[vertex] (da33) at  (6.5,-4.5) {$a_3$};
\node[vertex] (db33) at  (7.5,-4.5) {$b_3$};
\node[vertex] (da34) at  (4.5,-4.5) {$a_3$};
\node[vertex] (db34) at  (5.5,-4.5) {$b_3$};
\node at (0,-5.5) {\vdots};
\node at (6,-5.5) {\vdots};

%edges
\draw[edge] (a1) to (a2);
\draw[edge][bend left] (a1) to (a3);
\draw[edge][bend left] (a1) to (an);
\draw[edge] (a2) to (a3);
\draw[edge][bend left] (a2) to (an);
\draw[edge][bend left] (a3) to (an);
\draw[edge] (b1) to (b2);
\draw[edge][bend right] (b1) to (b3);
\draw[edge][bend right] (b1) to (bn);
\draw[edge] (b2) to (b3);
\draw[edge][bend right] (b2) to (bn);
\draw[edge][bend right] (b3) to (bn);
\draw[edge] (b1) to (a2);
\draw[edge] (b1) to (a3);
\draw[edge] (a1) to (b2);
\draw[edge] (a1) to (b3);
\draw[edge] (b2) to (a3);
\draw[edge] (a2) to (b3);
\draw[edgea] (a1) to (bna1);
\draw[edge] (a2) to (bn);
\draw[edgea] (a3) to (bna2);
\draw[edgea] (b1) to (anb1);
\draw[edge] (b2) to (an);
\draw[edgea] (b3) to (anb2);
\draw[edge] (da1) to (da21);
\draw[edge] (da1) to (db21);
\draw[edge] (db1) to (da22);
\draw[edge] (db1) to (db22);
\draw[edge] (da21) to (da31);
\draw[edge] (da21) to (db31);
\draw[edge] (db21) to (da32);
\draw[edge] (db21) to (db32);
\draw[edge] (da22) to (da34);
\draw[edge] (da22) to (db34);
\draw[edge] (db22) to (da33);
\draw[edge] (db22) to (db33);

\end{tikzpicture}
\end{center}
\caption{A digraph and its exponentially large valid and optimal DAG-depth decomposition without
merging of any vertices}
\label{gra:expo}
\end{figure}
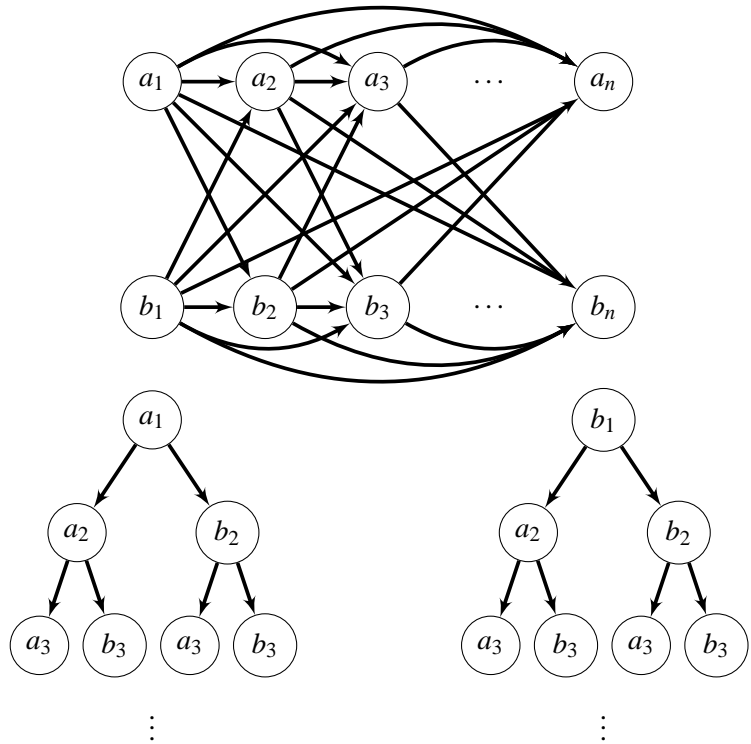

\begin{definition}
  Two vertices $u,v \in V(P)$ are \textit{mergeable} in the DAG-depth decomposition $(P,org)$
  if the conditions $1-3$ hold. Furthermore, $u$~and $v$ are \textit{optimally mergeable} if
  they are mergeable and also the condition $4$ holds.
  \begin{enumerate}
   \item $org(u)=org(v)$
   \item After merging $u$ with $v$, $P$ remains a DAG.
   \item Merging $u$ with $v$ does not break the Neighbor cover property from 
   Definition~\ref{def:nCover}.
   \label{con:mergingB}
   \item Merging $u$ with $v$ does not increase the depth of the decomposition.
   \label{con:mergingD}
  \end{enumerate}
\label{def:merging}
\end{definition}

For example, all duplicits in the example of Theorem~\ref{thm:expLarge} are optimally mergeable.

The following is obvious from the definition:

\begin{proposition}
  Let $(P,org)$ be a valid and optimal DAG-depth decomposition of some graph and $u,v \in V(P)$ an 
  optimally mergeable pair of vertices of $P$. Then, after merging $u$ and $v$, $P$ is still 
  a valid and optimal decomposition.
\label{thm:mergingvalid}
\end{proposition}

%Not following Definition~\ref{def:merging} and merging a pair of vertices which are not mergeable 
%will make the resulting decomposition no longer valid. The definition of the pairs of mergeable
%vertices only assures that after merging them the decomposition fulfills the definition of a 
%DAG-depth decomposition.

The trivial lower bound on the size of a valid decomposition is equal to the number of vertices of the 
original graph, but such a decomposition may not be optimal. The question of minimizing the size 
of a valid and optimal decomposition is left for further investigations.

\section{Closure of a DAG-depth decomposition}
\label{chp:closure}
While previous text focused on how to construct a DAG-depth decomposition, or a game plan, for a given
digraph, now we look from the other side. Roughly, having a game plan, can we easily say on which digraphs
we can win with it?

Recall that in the case of tree-depth this was trivial - already the definition of the tree-depth 
decomposition worked with the concept of a closure of a rooted forest, which, at the same time,
represents the unique maximal graph on which the cop player always wins when following the decomposition
as the game plan.

However, in the case of digraphs and DAG-depth, we again face unprecedented complications. A 
DAG-depth decomposition, unlike an elimination tree, can contain more copies of a single vertex
of the original digraph. Therefore a problem, which was trivial in undirected graphs, becomes complex.

In the closure obtained from an elimination tree, each vertex is connected with all of its former
ancestors and descendants. In a DAG-depth decomposition, more copies of a vertex can have different
ancestors and descendants.

We thus define the following.

\begin{definition}[Closure]
  A \textit{partial closure} $C$ is a directed graph obtained from a DAG-depth decomposition $(P,org)$ of
  some graph $D$, such that $D$ is a spanning subgraph of $C$ and $(P,org)$ is still a valid DAG-depth 
  decomposition for the digraph $C$. A \textit{closure} is a maximal partial closure by inclusion.
\label{def:closure}
\end{definition}

%A closure for a decomposition is not unique in general, see the example in Figure~\ref{gra:difrentclosures}.

\begin{theorem}
  For a DAG-depth decomposition $(P,org)$ of a digraph $D$, we construct a digraph $C$, such that 
  $V(C)=V(D)$ by iteratively adding edges $(u,v)$ for $u,v \in V(C)$ if for every $u' \in V(P)$ which
  is a copy of $u$
  \begin{enumerate}
    \item there exists $v' \in V(P)$ such that $v'$ is a copy of $v$ and $v'$ is a descendant of $u'$ 
    in $P$, or
    \item every path from a root of $P$ to $u'$ contains a copy of $v$.
  \end{enumerate}
  Then, $C$ is a closure of $P$, which is thus unique.
\label{thm:closure}
\end{theorem}
\begin{proof}
  The conditions in this theorem are the same as the Neighbor cover property in Definition~\ref{def:nCover}
  and so $C$ is clearly a partial closure. On the other hand, every other edge not in $E(C)$ would, by its
  own, violate Definition~\ref{def:nCover} and so $C$ is maximal.

\end{proof}

These are some of the informal ideas worth further investigation - see \cite{bakalarka} for more details.

\section{Conclusion}
\label{chp:end}
We have presented a definition of a DAG-depth decomposition which extends the concept of an elimination
tree as a winning strategy for the cop player in the lift-free version of the cops-and-robber game to directed graphs.
Unlike in the case of an elimination tree, the vertex set of a DAG-depth decomposition is not equal
to the vertex set of the original graph. That requires us to deal with the two vertex sets and to 
find a way to map between them. Since the vertex sets are not equal, the size aspect of the DAG-depth
decomposition becomes a problem. In the primitive handling, the size of the decomposition can grow 
exponentially. To make the decomposition smaller and therefore acceptable as a strategy for the cop
player, we provide a definition of mergeable and optimally mergeable vertices.

Secondly, we present a definition of the closure as the largest graph where the given decomposition 
works as a winning strategy. We also provide a way to deterministically obtain a closure for a 
given decomposition.

The main direction for possible future improvements and extension of our results is to study the lower
bounds on the size of a valid and optimal DAG-depth decomposition of a digraph
and a relationship between these bounds and the properties of given digraphs.

%
% ---- Bibliography ----
%
%\printbibliography
\bibliography{bib}
\end{document}